\documentclass[a4paper, pra,twocolumn,showpacs,superscriptaddress,groupedaddress]{revtex4}
\usepackage{ae}
\usepackage{physics}
\usepackage[T1]{fontenc}
\usepackage[ansinew]{inputenc}
\usepackage{amsmath}
\usepackage{amssymb}
\usepackage[caption=false]{subfig}
\usepackage{multirow}
\usepackage{array}
\usepackage[]{graphicx}
\usepackage{wrapfig}
\usepackage{makecell}
\usepackage{epstopdf}
\usepackage{mathtools}
\usepackage{color}
\usepackage[colorlinks=true,linkcolor=red,citecolor=blue,urlcolor=blue]{hyperref}
\usepackage{lscape}
\usepackage{amsthm}
\newtheorem{theorem}{Theorem}

\graphicspath{ {./image/} }

\newtheorem{definition}{Definition}
\begin{document}
\title{Catalytic Enhancement of Coherence in Noisy Quantum Channels and Characterization of Strictly Incoherent Operations}
\author{Priyabrata Char}
\email{priyabrata_ipdf@iitg.ac.in, mathpriyabrata@gmail.com}
\affiliation{Department of Physics, Indian Institute of Technology Guwahati,
Guwahati, 781039, Assam, India}
\author{Dipayan Chakraborty}
\email{dipayan@ssmahavidyalaya.org.in, dipayan.tamluk@gmail.com}
\affiliation{Department of Mathematics, Sukumar Sengupta Mahavidyalaya, Keshpur, Paschim Medinipur, 721150, West Bengal, India}
\author{Indrani Chattopadhyay}
\email{icappmath@caluniv.ac.in}
\affiliation{Department of Applied Mathematics, University of Calcutta, 92 A.P.C Road, Kolkata, 700009, West Bengal, India}
\author{Debasis Sarkar}
\email{dsarkar1x@gmail.com, dsappmath@caluniv.ac.in}
\affiliation{Department of Applied Mathematics, University of Calcutta, 92 A.P.C Road, Kolkata, 700009, West Bengal, India}
\begin{abstract}
In realistic quantum information processing tasks, quantum states are inevitably affected by environmental noise, leading to decoherence and degradation of useful quantum resources. The coherence fraction, which serves as an important figure of merit for several quantum protocols, may decrease significantly after the action of a noisy channel. Such degradation can result in unsatisfactory performance in real-world applications. In this work, we investigate whether catalysis can be used to pre-process the input state to enhance the coherence fraction of an output state from a quantum channel. Specifically, we study whether using a processed state $\rho_s'$ as the input to a quantum channel $\Lambda$, instead of the original state $\rho_s$, can yield an output state $\Lambda(\rho_s')$ whose coherence fraction exceeds that of $\Lambda(\rho_s)$. We analyze the conditions under which such an improvement is possible. We also provide a practical application of our setup through the phase discrimination task. Furthermore, we establish a necessary and sufficient condition for an incoherent state preserving CPTP(Completely Positive Trace Preserving) map $\mathcal{E}$ to be a particular type of Strictly Incoherent Operation (SIO). This characterization provides a new structural understanding of SIO and clarifies its role in coherence manipulation. Our results offer practical insights into coherence preservation and enhancement in noisy quantum processes and may be useful for optimizing quantum information protocols under realistic conditions. We also provide numerical examples to support our claims. 
\end{abstract}
\date{\today}
\pacs{ 03.67.Mn, 03.65.Ud.;}
\maketitle
\section{Introduction}
Quantum Resource Theory (QRT) \cite{chitambar2019quantum, Gour_2025} provides us a unified way to study various quantum phenomena such as entanglement \cite{ekert1991quantum, bennett1992communication,bennett1993teleporting, horodecki2009quantum}, coherence \cite{baumgratz2014quantifying, aberg2006quantifying,streltsov2017colloquium}, asymmetry \cite{marvian2013theory,marvian2016quantum,marvian2014extending}, imaginarity \cite{wu2021operational, wu2021resource,wu2024resource,biswas2026entanglement}, and quantum reference frames \cite{gour2008resource, gour2009measuring}, etc. Despite varying physical motivations in defining the elements of a specific QRT, many similarities can be found when analyzing the quantification and convertibility of such resources. A general QRT divides all quantum states into two categories, free states and resource states \cite{chitambar2019quantum, Gour_2025}. Additionally, a QRT also includes a set of quantum operations that emerge from natural constraints on the considered physical system, ensuring that these operations always produce a free state when applied to a free state. These quantum operations are called free operations \cite{chitambar2019quantum, Gour_2025}. The QRT then subsequently investigates quantum information processing protocols that can be implemented using free operations and their associated resource states. 

In the resource theory of entanglement, separable states are free states, Local Operation and Classical Communication (LOCC) \cite{chitambar2014everything, bhunia2023more} are free operations, and entangled states are resource states \cite{horodecki2009quantum}. In the resource theory of coherence, we fix a reference orthonormal basis $\ket{i}$ of the Hilbert space, say $\mathcal{H}_A$, which we call the incoherent basis. A state $\rho$ is said to be incoherent if it is diagonal in this basis and called a coherence state if it is not \cite{baumgratz2014quantifying}. Different classes of free operations exist for coherence theory, such as Maximally Incoherent Operations (MIO) \cite{aberg2006quantifying}, Incoherent Operations (IO) \cite{baumgratz2014quantifying}, Strictly Incoherent Operations (SIO) \cite{winter2016operational, yadin2016general}, and Physically Incoherent Operations (PIO) \cite{chitambar2016comparison,chitambar2017erratum,chitambar2016critical}. Each of them arises by considering different physical constraints on allowable transformations, along with the constraint that they cannot generate coherence from incoherent states. Quantifying resources is a pivotal task in any resource theory \cite{chitambar2019quantum}. Several important coherence measures are $l_1$ norm of coherence \cite{baumgratz2014quantifying}, distillable coherence \cite{winter2016operational,yuan2015intrinsic}, coherence cost \cite{winter2016operational}, coherence of formation \cite{winter2016operational}, relative entropy of coherence \cite{baumgratz2014quantifying}, robustness of coherence \cite{napoli2016robustness, piani2016robustness},  coherence fraction \cite{yao2019quantum, karmakar2019coherence,lipka2021catalytic} etc.

Quantum resource theories focus on quantifying and manipulating quantum resources through free operations to enhance their utilization and optimize the performance of quantum information tasks. But in practical implementation, a quantum system is inevitably exposed to environmental interactions, leading to decoherence and degradation of valuable quantum resources \cite{barnum1998information, schumacher1998quantum} like entanglement, coherence, etc.  Quantum channels are essential for modeling the effects of environmental noise on quantum states, such as depolarizing, dephasing, amplitude damping, and bit-flip, etc \cite{nielsen2010quantum,hayashi2006quantum,gyongyosi2018survey}. These channels typically reduce the resources, thereby affecting the efficiency of quantum protocols. Hence, the properties of quantum channels naturally comes under inspection within the framework of quantum resource theories. 

Mathematically, a quantum channel is a Completely Positive and Trace Preserving (CPTP) map $\Lambda: \mathcal{B}(\mathcal{H}_A) \to \mathcal{B}(\mathcal{H}_B) $. The completely positive property ensures that $\Lambda$ is a physical map whereas trace preserving ensures that the trace of the density matrix is preserved.  Any such quantum channel $(\Lambda)$ admits the following Kraus operator representation \cite{kraus1983states,hellwig1970operations,choi1975completely}
$$
\Lambda(\rho) = \sum_n K_n \rho K_n^\dagger, \quad \sum_n K_n^\dagger K_n = I,
$$
where $\{K_n\}$ are Kraus operators. Equivalently, quantum channels can also be described by the Stinespring dilation theorem \cite{stinespring1955positive}. The action of a quantum channel $\Lambda$ on a state $\rho$ can be described as joint unitary $(U)$ evolution of the main system and an environment system, followed by a partial trace over the environment system to the resultant state 
$$\Lambda(\rho)=\text{Tr}_E(U(\rho_s\otimes\ket{0}\bra{0}_E)U^{\dagger}).$$ This provides a clear physical interpretation of quantum channels as a unitary evolution of a larger, closed system consisting of the primary system plus an environment system, where we choose to ignore part of it.

The change in a state's resource content under a noisy quantum channel can be quantified by comparing the input and output state via a suitable resource quantifier \cite{takahashi2022creating, pal2014entanglement, mani2015cohering, kumar2025fidelity}. Naturally, the question arises whether it is possible to stop, or at least restrict, such resource degradation to some extent \cite{wakamura2017state, suter2016colloquium}. Understanding the behavior of the resource quantifier is crucial for designing strategies to protect the resource in the presence of noise. One promising solution may be provided by the use of quantum catalysis \cite{jonathan1999entanglement,aaberg2014catalytic,datta2023catalysis,lipka2024catalysis}. A catalyst is an auxiliary quantum state that can enable an impossible state transformation \cite{jonathan1999entanglement, bu2016catalytic, kondra2021catalytic, char2023catalytic, takagi2022correlation}, or enhance a quantum information processing task \cite{lipka2021catalytic, char2024boosting}, all while remaining unchanged and not being consumed in the process. Catalysis can be categorized as uncorrelated catalysis \cite{jonathan1999entanglement,aaberg2014catalytic,bu2016catalytic} and correlated catalysis \cite{char2024boosting,kondra2021catalytic,lipka2021catalytic, takagi2022correlation, shiraishi2021quantum} depending on the correlation between the post-transformation main system and the catalytic system. The uncorrelation constraint is a very powerful constraint, and in a practical scenario, such a condition may not always be fulfilled. The concept of correlated catalysis not only unlocks a broader set of achievable transformations than uncorrelated catalysis in quantum resource theories but also offers a more ideal framework for practical experiments \cite{char2024boosting,kondra2021catalytic,lipka2021catalytic, takagi2022correlation, shiraishi2021quantum,datta2024entanglement,bavaresco2025catalytic,lipka2025finite,ding2021amplifying,zhang2024experimental}.

It is natural to investigate whether catalytic techniques can be useful to enhance or preserve quantum coherence in the presence of noise. In this work, we explore when catalytic techniques can be employed to preprocess input states so that the processed state, after passing through a noisy quantum channel, exhibits an enhanced coherence fraction. We specifically consider coherence fraction because it often arises in various quantum information protocols, where the protocol's performance metric \cite{napoli2016robustness, karmakar2023coherence, grondalski2002fully, bu2017maximum, horodecki1999general} is associated with it. In the past, it was also observed that catalysis processes have a positive influence on such quantum information protocols \cite{kondra2021catalytic, char2024boosting}. Section \ref{sec:Priliminaries} contains the necessary preliminaries. In section \ref{sec:catalytic advantage}, we have discussed how and when catalysis can increase coherence fraction in the presence of a noisy channel and an application of this concept to a quantum information protocol. In section \ref{sec:Example} and \ref{sec:example 1} we have provided two examples of our catalytic scheme. Section \ref{sec:SIO} contains a necessary and sufficient criterion for an MIO to be an SIO. Finally, we conclude our findings in section \ref{sec:conclusion}.

\section{Preliminaries}
\label{sec:Priliminaries}
To start, we will explore several key coherence metrics that are essential for this work's context.

$l_1$ norm of coherence measures the total amount of coherence present in the state \cite{baumgratz2014quantifying} and defined as 
\begin{equation}
    C_{l_1}(\rho)= \min_{\sigma\in \mathcal{I}}\parallel\rho-\sigma\parallel_{l_1}=\sum_{i\neq j}|\rho_{ij}|.
    \label{eq:l1 norm}
\end{equation}
where $\mathcal{I}$ is the set of incoherent states.

Another measure,  the robustness of coherence \cite{napoli2016robustness,piani2016robustness}, assesses the extent to which an incoherent state $\sigma$ must be combined with a specified state $\rho$ to remove its coherence completely. The larger the robustness of coherence, the more the system can withstand being changed into an incoherent state through free operations. It is defined as 
\begin{equation}
    C_R(\rho)=\min_{\sigma\in \mathcal{I}}\{x\geq0|\frac{\rho+x\sigma}{1+x}\in \mathcal{I}\}.
    \label{eq:robustness of coherence}
\end{equation}
%The robustness of coherence plays a significant role in various coherence manipulation tasks like phase discrimination \cite{phase_discrimination1, phase_discrimination2} and sub channel discrimination \cite{subchannel_discrimination} game.
The coherence fraction measures the maximum overlap between a quantum state and a maximal coherent state \(\ket{\phi}\) \cite{yao2019quantum, karmakar2019coherence,lipka2021catalytic} within a fixed basis. In the case of a single-party system, it is defined as 
\begin{equation}
    C_F(\rho)=\max_{U} \bra{\phi}U\rho U^{\dagger}\ket{\phi}=\max_{\Lambda\in \text{IO}}\bra{\phi}\Lambda(\rho)\ket{\phi},
    \label{eq:coherence fraction for single party}
\end{equation}
where $U$ includes all unitary operations that are incoherent. 
%For a bipartite (or distributed) scenario, let $\ket{\chi}$ denote the maximally coherent state, and $\bar{U}=U_A\otimes U_B$, where $U_A,\; U_B$ are local incoherent unitary operations. In this situation, the coherence fraction \cite{coherence_fraction2,catalytic_teleportation} can be defined as
%\begin{equation}
%    C_F(\rho)=\max_{\bar{U}} \bra{\chi}\bar{U}\rho \bar{U}^{\dagger}\ket{\chi}=\max_{\bar{\Lambda}\in \text{LICC}}\bra{\chi}\bar{\Lambda}(\rho)\ket{\chi}.
 %   \label{eq:coherence fraction for multi party}
%\end{equation}
For maximally coherent states in $d$ dimensions \cite{baumgratz2014quantifying}, the coherence fraction equals $1$, while for any incoherent state of dimension $d$, it is $1/d$ \cite{yao2019quantum, karmakar2019coherence}. 

The next theorem states the close connection of coherence fraction in a single-party scenario with $l_1$ norm coherence and robustness of coherence.

\begin{theorem}
   \cite{yao2019quantum} For a $d$- dimensional state $\rho$ 
   \begin{equation}
  C_{F}(\rho)\leq\frac{1+C_R(\rho)}{d}  ,
  \label{eq:fraction with robustness inequality}
\end{equation}  
and if  $\rho$ can be transferred to $\Tilde{\rho}=U^{\dagger}\rho U$ with entries $\tilde{\rho}_{ij}=|\rho_{ij}|$, where $U$ is an incoherent unitary operator, then 
\begin{equation}
  C_{F}(\rho)=\frac{1+C_{l_1}(\rho)}{d}=\frac{1+C_R(\rho)}{d}.  
  \label{eq:fraction with robustness}
\end{equation} 
\end{theorem}
However, for a 2-dimensional state $\rho$, the following relation holds always,\cite{karmakar2019coherence}
    \begin{equation}
        C_{F}(\rho)=\frac{1+C_R(\rho)}{2}.
        \label{eq:fraction with robustness for qubit}
    \end{equation} 
Next, we will discuss various free operations of coherence resource theory.  Several operational classes are commonly considered, each imposing a different level of restriction on the admissible operations.

\textbf{Maximally Incoherent Operations}: A completely positive trace-preserving map $\mathcal{E}$ (CPTP map) is called a Maximally Incoherent Operation (MIO) \cite{aberg2006quantifying} if it maps every incoherent state to an incoherent state, i.e.,

\begin{equation}
\mathcal{E}(\delta) \in \mathcal{I} \qquad \text{for all } \delta \in \mathcal{I},
\end{equation}
where \(\mathcal{I}\) denotes the set of incoherent states.

\textbf{Incoherent Operations}: A completely positive trace-preserving map $\mathcal{E}$ is called an Incoherent Operation (IO) \cite{baumgratz2014quantifying} if it admits a Kraus representation
\begin{equation}
\Lambda(\rho) = \sum_n K_n \rho K_n^\dagger, \quad \sum_n K_n^\dagger K_n=I.
\end{equation}
such that each Kraus operator $K_n$ individually preserves incoherence state,
\begin{equation}
K_n \mathcal{I} K_n^\dagger \subseteq \mathcal{I}.
\end{equation}
Each Kraus operator $K_n$ for IO has at most one non-zero element in each column.

\textbf{Strictly Incoherent Operations}: A CPTP map $\mathcal{E}$ is said to be a Strictly Incoherent Operation (SIO) \cite{winter2016operational, yadin2016general} if it admits a Kraus representation 
\begin{equation}
\Lambda(\rho) = \sum_n K_n \rho K_n^\dagger, \quad \sum_n K_n^\dagger K_n=I.
\end{equation}
such that both $K_n$ and $K_n^\dagger$ preserve incoherent states. That is, 
\begin{equation}
    K_n \mathcal{I} K_n^\dagger \subseteq \mathcal{I},
    \qquad 
    K_n^\dagger \mathcal{I} K_n \subseteq \mathcal{I}.
\end{equation}
In the case of an SIO, each Kraus operator $K_n$ has at most one nonzero entry in every row and every column.

Next, we will define a quantum channel important to this work and a theorem related to its Kraus operator structure.

\textbf{Schur Multiplier Channel}: A linear map $\mathcal{E}: \mathcal{B(H)}\rightarrow\mathcal{B(H)}$ is called a Schur multiplier channel with respect to a fixed orthonormal basis $\{\ket{i}\}_{i=1}^d$ if it acts entrywise as 
$$\mathcal{E}(X)=C \odot X=[c_{ij}x_{ij}]_{d\times d}$$
where $X=[x_{ij}]_{d \times d} \in \mathbb{C}^{d \times d}$ and $C=[c_{ij}]_{d \times d} \in \mathbb{C}^{d \times d}$ and $C \succeq 0$ with $C_{ii}=1$ for all $i$ \cite{Watrous2018}.

\begin{theorem}
\label{th:Dipayan's 1st theorem}
A quantum channel $\mathcal{E}$ is a Schur multiplier channel in a given basis if and only if it admits a Kraus decomposition 
$$\mathcal{E}(X)=\sum_n K_n X K_n^\dagger,\quad \sum_n K_n^\dagger K_n=I,$$
where each Kraus operator $K_n$ is diagonal \cite{Watrous2018,Levick2018} in that basis and from this representation it is obvious that $\mathcal{E}$ is a SIO.
\end{theorem}
Note that every Schur multiplier channel is an SIO, but not every SIO is a Schur multiplier channel.

%\textbf{Physically Incoherent Operations}: A CPTP map $\mathcal{E}$ is called a Physically Incoherent Operation (PIO) \cite{chitambar2016comparison,chitambar2017erratum,chitambar2016critical} if it can be implemented using only incoherent ancilla states, incoherent unitaries, incoherent measurements and classical post processing. Equivalently, PIO consists of channels whose kraus operators arise from such physical implementations.

   % PIO is defined by a CPTP map which can be expressed as a convex combination of maps each having Kraus operators $K_j$ of the form 

  %  $$K_j=U_jP_j=\sum_{x} e^{i \theta_x} \ket{\pi_j (x)}\bra{x} P_j$$
    
%Quantum channels are linear maps between density operators that preserves trace and positivity. Formally, a map $\Lambda: \mathcal{B}(H_A) \rightarrow \mathcal{B}(H_B)$ is a quantum channel if it is CPTP. Physical realizability is ensured by Complete Positivity while total probability normalized is kept by Trace Preserving condition.

%Any quantum channel admits a decomposition $\Lambda(\rho)=\sum_{n} K_n \rho K_n^\dagger$, where $\{K_n\}$ are Kraus operators satisfying $\sum_n K_n^\dagger K_n=I$. The operator sum form arises from following dilation \\
%The system interacts unitarily with an environment before tracing out the latter. \\
%Some important quantum channels are the Depolarizing channel, the Phase damping channel, the Bit-flip channel, the Amplitude damping channel, and the Quantum addition channel.

\section{Catalytic advantage in single-party scenario}
\label{sec:catalytic advantage}
In this section, we will first define the catalytic IO and then state a theorem that establishes a connection between catalytic and asymptotic IO transformations. Later in this section, we will present our two key results regarding catalysis and the quantum channel. 
\subsection{Catalytic IO transformation}
\begin{definition}
    Consider a single-party system `\text{s}' and a catalytic system c. A catalytic incoherent operation (IO) on a state $\rho_s$ of this system is defined as $\rho_s\rightarrow \lim\limits_{n\to\infty} tr_c [\mathcal{E}_n(\rho_s\otimes \tau_{c}^{n})]$, with the following conditions \\
    (i) $\tr_s [\mathcal{E}_n(\rho_s\otimes \tau_{c}^{n})]=\tau_{c}^{n}, \forall n$ and\\ (ii)$\lim\limits_{n\to\infty}\parallel\mu_{sc}^{n}-\sigma_{s}\otimes\tau_{c}^{n}\parallel_1=0$,\\ where $\{\tau_{c}^{n}\}$ is a sequence of catalyst states, $\{\mathcal{E}_n\}$ is a sequence of IO,  $\mu_{sc}^{n}=\mathcal{E}_n(\rho_s\otimes \tau_{c}^{n})$, and $\sigma_s$ is the target state where Alice possesses both $s$ and $c$ \cite{char2023catalytic,char2024boosting}.
\end{definition}  
\begin{theorem}
\label{th:catalytic IO}
     If it is possible to convert a single-party state $\rho_{s}$ into $\sigma_s$ using asymptotic IO at a unit rate, then it follows that there is a catalytic IO capable of transforming $\rho_{s}$ into $\sigma_s$ \cite{char2023catalytic,char2024boosting}.
\end{theorem} 
Let Alice (A) possess the whole system $s$. As state conversion from $\rho_s$ to $\sigma_s$ is possible via asymptotic IO with unit rate, then depending on arbitrarily chosen $\epsilon>0, \exists n\in \mathbb{N}$ and an IO transformation $\mathcal{E}$ such that $\mathcal{E}[(\rho_s)^{\otimes n}]=\Gamma$ with $D(\Gamma,(\sigma_s)^{\otimes n})<\epsilon$ where $D$ is a trace-1 norm.
Considering the catalytic state
$$\tau_c=\frac{1}{n}\sum_{k=1}^{n}\underbrace{\rho^{\otimes (k-1)}\otimes\Gamma_{n-k}}_{A}\otimes \underbrace{\ket{k}\bra{k}}_{R},$$
where $\Gamma_i=\text{tr}_{A_1A_2...A_{n-i}}(\Gamma)$, $\Gamma_0=I$ and $c=A\otimes R $ then,
$$\rho_s\otimes\tau_c=\frac{1}{n}\sum_{k=1}^{n}\rho^{\otimes k}\otimes\Gamma_{n-k}\otimes \ket{k}\bra{k}.$$
Next, we describe the catalytic IO process in three steps.

$\mathbf{Step-1:}$ Alice carries out a rank 1 projective measurement using the basis $\ket{k}$. If Alice receives the result $n$, she will perform the IO operation $\mathcal{E}$, if not, she will take no action, and the resulting state will be
$$\mu_{sc}'=\frac{1}{n}\sum_{k=1}^{n-1}\rho^{\otimes k}\otimes\Gamma_{n-k}\otimes\ket{k}\bra{k}+\frac{1}{n}\Gamma\otimes\ket{n}\bra{n}.$$
$\mathbf{Step-2:}$ Alice subsequently applies an incoherent unitary operation to the auxiliary system of $\mu_{sc}'$, that sends $\ket{i}$ to $\ket{i+1}$ and $\ket{n}$ to $\ket{1}$. Her resultant state will be
$$\mu_{sc}''=\frac{1}{n}\sum_{k=1}^{n}\rho^{\otimes (k-1)}\otimes\Gamma_{n+1-k}\otimes\ket{k}\bra{k}.$$
$\mathbf{Step-3:}$ Alice now employs a swap unitary that sends $s_i$ to $s_{i+1}$ and $s_n$ to $s_1$, after which $\mu_{sc}''$ transforms to the final state $\mu_{sc}$, such that $D(\mu_{sc},\sigma_s\otimes\tau_c)<2\epsilon$ and $\text{tr}_s(\mu_{sc})=\tau_c$ and 
\begin{equation}
    \text{tr}_c(\mu_{sc})=\frac{1}{n}\sum_{i=1}^{n}\text{tr}_{/i}\mathcal{E}(\rho_s^{\otimes n})=\rho_{s}',
    \label{eq:final state}
\end{equation} where $\text{tr}_{/i}(\rho_{s_1s_2s_3...s_i...s_n})=\rho_{s_i}$.\\
Later in this manuscript, we call this $\rho_s'$ in equation \eqref{eq:final state} the processed state.
\subsection{Increasing the coherence fraction of the output state from a quantum channel.}
In many quantum information processing tasks, one particle must pass through a quantum channel $\Lambda$, which can introduce impurities into the state, resulting in changes to the coherence fraction of the output state compared to the input state $\rho_s$. This may lead to a problem, as the coherence fraction acts as a figure of merit for many protocols. A lower coherence fraction can lead to non-satisfactory results or such protocols. We now investigate whether using the processed state $\rho_s'$ (see equation \eqref{eq:final state}) made by the catalytic process in theorem \ref{th:catalytic IO} as the input to the channel $\Lambda$, instead of the original state $\rho_s$, can result in an output state $\Lambda(\rho_s')$ whose coherence fraction exceeds that of the original output state $\Lambda(\rho_s)$.

Here we consider the case where the coherence fraction of the output state $\Lambda(\rho_s)$ from channel $\Lambda$ decreases compared to the coherence fraction of the input state $\rho_s$, that is
\begin{equation}
    \label{eq:3rd inequality}
    C_F(\rho_s)\geq C_F(\Lambda(\rho_s)).
\end{equation}
\begin{figure}[ht]
\centering
\includegraphics[width=0.45\textwidth]{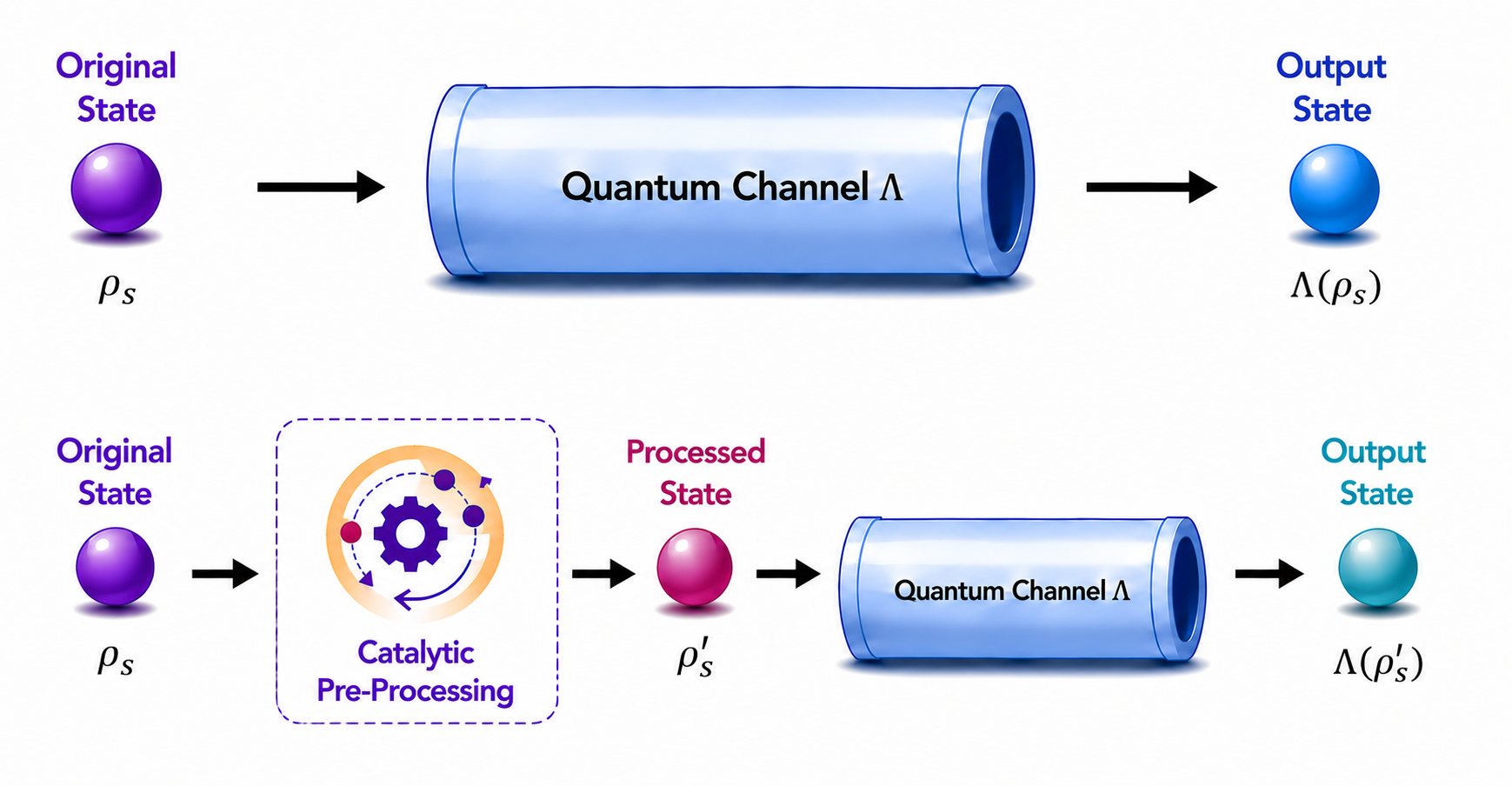}
\caption{Schematic diagram of the setup for this work}
\label{fig:channel picture}
\end{figure}
\begin{theorem}
\label{th:main 2nd theorem}
   By the catalytic process in theorem \ref{th:catalytic IO}, we have transferred the original input state $\rho_s$ to a processed state $\rho_s'$ (see equation \eqref{eq:final state}). If for a channel $\Lambda$ and IO $\mathcal{E}$, $C_F(\Lambda(\mathcal{E}(\rho_s)))=C_F(\mathcal{E}(\Lambda(\rho_s)))$ then 
   \begin{equation}
    \label{eq:main inequality}
    \lim_{n\to \infty}\max_{\mathcal{E}\in\text{IO}}C_{F}(\Lambda(\rho_s'))\geq C_{F}(\Lambda(\rho_s)).
\end{equation}
\end{theorem}
\begin{proof}
$$\lim_{n\to \infty}\max_{\mathcal{E}\in\text{IO}}C_{F}(\Lambda(\rho_s'))=\lim_{n\to \infty} \max_{\mathcal{E}\in\text{IO}}\max_{\Omega\in\text{IO}}\bra{\phi^{+}_{s}}\Omega(\Lambda(\rho_s'))\ket{\phi^{+}_{s}},$$
$$  \geq  \lim_{n\to \infty} \max_{\Omega\in\text{IO}}\bra{\phi^{+}_{s}}\Omega(\Lambda(\rho_s'))\ket{\phi^{+}_{s}},$$
$$ =  \lim_{n\to \infty} \max_{\Omega\in\text{IO}}\bra{\phi^{+}_{s}}\Omega\left(\Lambda\left(\frac{1}{n}\sum_{i=1}^{n}\text{tr}_{/i}\mathcal{E}(\rho_s^{\otimes n})\right)\right)\ket{\phi^{+}_{s}},$$
$$ =  \lim_{n\to \infty} \max_{\Omega\in\text{IO}}\bra{\phi^{+}_{s}}\frac{1}{n}\sum_{i=1}^{n}\Omega\left(\Lambda\left(\text{tr}_{/i}\mathcal{E}(\rho_s^{\otimes n})\right)\right)\ket{\phi^{+}_{s}},$$
\begin{center}
    as $\Lambda$, $\Omega$ are CPTP maps they are linear,
\end{center}
$$ =  \lim_{n\to \infty} \max_{\Omega'\in\text{IO}}\bra{\phi^{+}_{s}}\frac{1}{n}\sum_{i=1}^{n}\text{tr}_{/i}(\Omega'(\Lambda'(\mathcal{E}(\rho_s^{\otimes n}))))\ket{\phi^{+}_{s}},$$
$$\textrm{where, }\Omega'=\underbrace{\Omega\otimes\Omega\otimes ...\otimes\Omega}_{n  \text{ times}}, \;\text{and } \Lambda'=\underbrace{\Lambda\otimes\Lambda\otimes...\otimes\Lambda}_\textrm{$n$ times},$$
\begin{equation}
    \label{eq:commute 1st step}
     =\lim_{n\to \infty} \max_{\Omega'\in\text{IO}}\frac{1}{n}\sum_{i=1}^{n}\bra{\phi^{+}_{s}}\text{tr}_{/i}(\Omega'(\Lambda'(\mathcal{E}(\rho_s^{\otimes n}))))\ket{\phi^{+}_{s}},
\end{equation}
Using $C_F(\Lambda'(\mathcal{E}(\rho)))=C_F(\mathcal{E}(\Lambda'(\rho)))$ then
\begin{equation}
    \label{eq:commute 2nd step}
    =  \lim_{n\to \infty} \max_{\Omega'\in\text{IO}}\frac{1}{n}\sum_{i=1}^{n}\bra{\phi^{+}_{s}}\text{tr}_{/i}(\Omega'(\mathcal{E}(\Lambda'(\rho_s^{\otimes n}))))\ket{\phi^{+}_{s}},
\end{equation}
$$ =  \lim_{n\to \infty} \max_{\Omega''\in\text{IO}}\frac{1}{n}\sum_{i=1}^{n}\bra{\phi^{+}_{s}}\text{tr}_{/i}(\Omega''((\Lambda(\rho_s))^{\otimes n}))\ket{\phi^{+}_{s}},$$
where, $\Omega''=\Omega'o \text{ }\mathcal{E}$ is the composed IO
$$=\lim_{n\to \infty} \frac{C_{F_n}((\Lambda(\rho_{s}))^{\otimes n})}{n}=C_{F}^{\text{\text{reg}}}(\Lambda(\rho_s)).$$
\begin{equation}
\label{eq:1st inequality}
\text{Therefore }
    \lim_{n\to \infty}\max_{\mathcal{E}\in\text{IO}}C_{F}(\Lambda(\rho_s'))\geq C_{F}^{\text{\text{reg}}}(\Lambda(\rho_s))
\end{equation}
$$C_{F_n}((\Lambda(\rho_{s}))^{\otimes n})=\max_{\Omega''\in \text{IO}}\sum_{i=1}^{n}\bra{\phi^{+}_{s}}\text{tr}_{/i}(\Omega''((\Lambda(\rho_s))^{\otimes n}))\ket{\phi^{+}_{s}}$$
If we consider a special case where $\Omega''=\Omega''_1\otimes\Omega''_2\otimes...\otimes\Omega''_n$ with $\Omega''_1=\Omega''_2=...=\Omega''_n$ we get from the above 
$$C_{F_n}((\Lambda(\rho_{s}))^{\otimes n})\geq\max_{\Omega''_1\in \text{IO}}\sum_{i=1}^{n}\bra{\phi^{+}_{s}}\text{tr}_{/i}(\Omega''_1((\Lambda(\rho_s))^{\otimes n}))\ket{\phi^{+}_{s}}$$
$$=n\max_{\Omega''_1\in \text{IO}}\bra{\phi^{+}_{s}}\Omega''_1(\Lambda(\rho_s))\ket{\phi^{+}_{s}}$$
$$=n\;C_{F}(\Lambda(\rho_s))$$
$$\text{therefore } \frac{ C_{F_n}((\Lambda(\rho_{s}))^{\otimes n})}{n}\geq C_{F}(\Lambda(\rho_s))$$
Taking the limit $n\to\infty$ on both sides, we get 
\begin{equation}
\label{eq:2nd inequality}
    C_{F}^{\text{\text{reg}}}(\Lambda(\rho_s))\geq C_{F}(\Lambda(\rho_s))
\end{equation}
Combining equation \eqref{eq:1st inequality} and \eqref{eq:2nd inequality} we get 
\begin{equation}
%\label{eq:main inequality 3}
    \lim_{n\to \infty}\max_{\mathcal{E}\in\text{IO}}C_{F}(\Lambda(\rho_s'))\geq C_{F}(\Lambda(\rho_s)).
\end{equation}
\end{proof}
We next have to focus on the cases where $C_F(\Lambda(\mathcal{E}(\rho_s)))=C_F(\mathcal{E}(\Lambda(\rho_s)))$ is valid. The most obvious answer will be when the channel $\Lambda$ and the IO $\mathcal{E}$ commute with each other. In section \ref{sec:Example}, we give an example of a quantum channel that commutes with a SIO. Algebraically, two CPTP maps will commute with each other if and only if their respective super-operator matrices commute. Any CPTP $\Lambda$ may can be written as Kraus decomposition $\Lambda(\rho)=\sum_nK_n\rho K_{n}^{\dagger}$ where $\sum K_{n}^{\dagger}K_{n}=I.$ The super-operator of the CPTP map $\Lambda$ can be computed by $\sum_n K_n\otimes K_{n}^{*}.$ 

Apart from the commutative condition, there may exist channel and IO which do not commute, but $C_F(\Lambda'(\mathcal{E}(\rho)))=C_F(\mathcal{E}(\Lambda'(\rho)))$. As an example, let us consider the completely dephasing channel, which is defined as 
\begin{equation}
    \Delta(\rho)=\sum_{i=0}^{d-1}\ket{i}\bra{i}\rho\ket{i}\bra{i}.
\end{equation}
Clearly, for this channel $$C_F(\Delta(\mathcal{E}(\rho)))=C_F(\mathcal{E}(\Delta(\rho)))=\frac{1}{d}.$$
Hence, for a completely dephasing channel $\Delta$, the inequality \eqref{eq:main inequality} saturates, and we can not have an advantage for using the processed state.

But if $C_F(\Lambda(\mathcal{E}(\rho_s)))\neq C_F(\mathcal{E}(\Lambda(\rho_s)))$  what will happen ? Can we still use the catalytic pre-processing strategy, and if so, under what conditions? The next theorem specifically answers these questions.
\begin{theorem}
\label{th:1st main theorem}
     By the catalytic process in theorem \ref{th:catalytic IO}, we have transferred the original input state $\rho_s$ to a processed state $\rho_s'$ (see equation \eqref{eq:final state}). If for a channel $\Lambda$ and IO $\mathcal{E}$, $C_F(\Lambda(\mathcal{E}(\rho_s)))\neq C_F(\mathcal{E}(\Lambda(\rho_s)))$, but $\Lambda$ is an incoherent channel then 
\begin{equation}
    \label{eq:main inequality 3}
    \lim_{n\to \infty}\max_{\mathcal{E}\in\text{IO}}C_{F}(\Lambda(\rho_s'))\geq C_{F}(\Lambda(\rho_s)).
\end{equation}
\end{theorem}
\begin{proof}
    $$\lim_{n\to \infty}\max_{\mathcal{E}\in\text{IO}}C_{F}(\Lambda(\rho_s'))=\lim_{n\to \infty} \max_{\mathcal{E}\in\text{IO}}\max_{\Omega\in\text{IO}}\bra{\phi^{+}_{s}}\Omega(\Lambda(\rho_s'))\ket{\phi^{+}_{s}},$$
$$  \geq  \lim_{n\to \infty} \max_{\Omega\in\text{IO}}\bra{\phi^{+}_{s}}\Omega(\Lambda(\rho_s'))\ket{\phi^{+}_{s}},$$
$$ =  \lim_{n\to \infty} \max_{\Omega\in\text{IO}}\bra{\phi^{+}_{s}}\Omega\left(\Lambda\left(\frac{1}{n}\sum_{i=1}^{n}\text{tr}_{/i}\mathcal{E}(\rho_s^{\otimes n})\right)\right)\ket{\phi^{+}_{s}},$$
$$ =  \lim_{n\to \infty} \max_{\Omega\in\text{IO}}\bra{\phi^{+}_{s}}\frac{1}{n}\sum_{i=1}^{n}\Omega\left(\Lambda\left(\text{tr}_{/i}\mathcal{E}(\rho_s^{\otimes n})\right)\right)\ket{\phi^{+}_{s}},$$
\begin{center}
    as $\Lambda$, $\Omega$ are CPTP maps they are linear,
\end{center}
$$ =  \lim_{n\to \infty} \max_{\Omega'\in\text{IO}}\bra{\phi^{+}_{s}}\frac{1}{n}\sum_{i=1}^{n}\text{tr}_{/i}(\Omega'(\Lambda'(\mathcal{E}(\rho_s^{\otimes n}))))\ket{\phi^{+}_{s}},$$
$$\textrm{where, }\Omega'=\underbrace{\Omega\otimes\Omega\otimes ...\otimes\Omega}_{n  \text{ times}}, \;\text{and } \Lambda'=\underbrace{\Lambda\otimes\Lambda\otimes...\otimes\Lambda}_\textrm{$n$ times},$$
$$ =  \lim_{n\to \infty} \max_{\Omega'\in\text{IO}}\frac{1}{n}\sum_{i=1}^{n}\bra{\phi^{+}_{s}}\text{tr}_{/i}(\Omega'(\Lambda'(\mathcal{E}(\rho_s^{\otimes n}))))\ket{\phi^{+}_{s}},$$
If the channel $\Lambda$ is an Incoherent channel, that is, the Kraus matrix has at most one non-zero element in each column, then $\Omega'''=\Omega'o\text{ }\Lambda\text{ } o\text{ }\mathcal{E}$ is also an IO.
$$ =  \lim_{n\to \infty} \max_{\Omega'''\in\text{IO}}\frac{1}{n}\sum_{i=1}^{n}\bra{\phi^{+}_{s}}\text{tr}_{/i}(\Omega'''(\rho_s^{\otimes n}))\ket{\phi^{+}_{s}},$$
$$=\lim_{n\to \infty} \frac{C_{F_n}(\rho_{s}^{\otimes n})}{n}=C_{F}^{\text{\text{reg}}}(\rho_s).$$
\begin{equation}
\label{eq:4th inequality}
\text{Therefore }
    \lim_{n\to \infty}\max_{\mathcal{E}\in\text{IO}}C_{F}(\Lambda(\rho_s'))\geq C_{F}^{\text{\text{reg}}}(\Lambda(\rho_s))
\end{equation}
$$C_{F_n}(\rho_{s}^{\otimes n})=\max_{\Omega'''\in \text{IO}}\sum_{i=1}^{n}\bra{\phi^{+}_{s}}\text{tr}_{/i}(\Omega'''(\rho_s^{\otimes n}))\ket{\phi^{+}_{s}}$$
If we consider a special case where $\Omega'''=\Omega'''_1\otimes\Omega'''_2\otimes...\otimes\Omega'''_n$ with $\Omega'''_1=\Omega'''_2=...=\Omega'''_n$ we get 
$$C_{F_n}(\rho_{s}^{\otimes n})\geq\max_{\Omega'''_1\in \text{IO}}\sum_{i=1}^{n}\bra{\phi^{+}_{s}}\text{tr}_{/i}(\Omega'''_1(\rho_s^{\otimes n}))\ket{\phi^{+}_{s}}$$
$$=n\max_{\Omega'''_1\in \text{IO}}\bra{\phi^{+}_{s}}\Omega'''_1(\rho_s)\ket{\phi^{+}_{s}}$$
$$=n\;C_{F}(\rho_s)$$
$$\text{therefore } \frac{ C_{F_n}(\rho_{s}^{\otimes n})}{n}\geq C_{F}(\rho_s)$$
Taking the limit $n\to\infty$ on both sides, we get 
\begin{equation}
\label{eq:5th inequality}
    C_{F}^{\text{\text{reg}}}(\rho_s)\geq C_{F}(\rho_s)
\end{equation}
Combining equation \eqref{eq:3rd inequality}, \eqref{eq:4th inequality} and \eqref{eq:5th inequality} we get 
\begin{equation}
\label{eq:main inequality 2}
    \lim_{n\to \infty}\max_{\mathcal{E}\in\text{IO}}C_{F}(\Lambda(\rho_s'))\geq C_{F}(\rho_s)\geq C_{F}(\Lambda(\rho_s)).
\end{equation}
\end{proof}

In section \ref{sec:example 1}, for the qutrit phase damping channel, we provide a numerical example of the advantage achieved by using a catalysis.

Note that the inequality \eqref{eq:main inequality 2} is violated for the complete dephasing channel $\Delta$, which is defined as 
\begin{equation}
    \Delta(\rho)=\sum_{i=0}^{d-1}\ket{i}\bra{i}\rho\ket{i}\bra{i}.
\end{equation}
For this channel, the left-most expression and also the right most expressions' value will be $\frac{1}{d}$, and in general, if the input state $\rho_s$ has some non-zero coherence, then $C_F(\rho_s)>\frac{1}{d}$. 

\subsection{Application in phase discrimination}
In phase discrimination  \cite{napoli2016robustness} Alice holds a  $d$ dimensional quantum state $\rho_s$ on which phase encoding is implemented through unitary action $U_{\phi}=exp(iH\phi)$ with $H=\sum_{j=0}^{d-1}j\ket{j}\bra{j}$ and $\phi\in\mathbb{R}$. The basis for coherence is determined by the evenly spaced spectrum of the Hamiltonian $H$ of the system. After the phase encoding the output state is $f_{\phi}(\rho_s)=U_{\phi} \rho_s U_{\phi}^{\dagger}$. If one of the $n$ phases $\{\phi_{k}|k=1,1,2,...,n\}$ are encoded with probability $p_k$, the collection of pairs 
$\Theta=\{(p_k,\phi_k)|k=1,2,...,n\}$ defines a phase discrimination game. Alice's objective is to accurately determine the encoded phase by carrying out a generalized measurement on $f_{\phi}(\rho_s)$ using measurement operators $\{\mathcal{M}_k\}$ that satisfy $\mathcal{M}_k \geq 0$ and $\sum_k \mathcal{M}_k = I$.  The optimal probability  of successfully determine the encoded phase  \cite{napoli2016robustness} can be calculated as,
\begin{equation}
    P^{\text{succ}}_{\Theta}(\rho_s)=\max_{\{\mathcal{M}_k\}} \sum_k p_k \text{tr}[U_{\phi_k} \rho_s U_{\phi_k}^{\dagger}\mathcal{M}_k ].
\end{equation}
On the other hand, if Alice's possesses an incoherent state $\delta_s$, the probability of successfully determining the encoded phase will be \cite{napoli2016robustness}
\begin{equation}
    P^{\text{succ}}_{\Theta}(\delta_s)=\max_{\{k\}} p_k.
\end{equation}
The maximum benefit obtained from using a coherence state $\rho_s$ instead of any incoherent state $\delta_s$ in all possible phase discrimination scenarios can be precisely calculated by the robustness of coherence of $\rho_s$ \cite{napoli2016robustness},
\begin{equation}
    \max_{\Theta}\frac{p_{\Theta}^{\text{succ}}(\rho_s)}{p_{\Theta}^{\text{succ}}(\delta_s)}=1+C_{R}(\rho_s).
    \label{eq:maximum advantage}
\end{equation}
We call the L.H.S. of equation \eqref{eq:maximum advantage} as the \textit{maximum advantage ratio} in phase discrimination for a state $\rho_s$. 

Thus, the slightest amount of coherence in the input state can help Alice better determine the actual implemented phase, making quantum coherence a vital resource for phase discrimination. The more coherence there is, the better it is for the task. However, in realistic scenarios, noise or decoherence (modeled mathematically by a quantum channel $\Lambda$) typically degrades the coherence of the input state for the phase discrimination task. In a realistic scenario, the input state for phase discrimination is $\Lambda(\rho_s)$ instead of $\rho_s$, and the \textit{maximum advantage ratio} reduces as typically $C_R(\Lambda(\rho_s))< C_R(\rho_s).$ It would be better if we could somehow preserve the coherence of the input state. In the next few paragraphs, using theorems \ref{th:main 2nd theorem} and \ref{th:1st main theorem}, we will prove that using a correlated catalysis, we can actually increase the \textit{maximum advantage ratio} of the phase discrimination task. 
\begin{figure}[ht]
\centering
\includegraphics[width=0.45\textwidth]{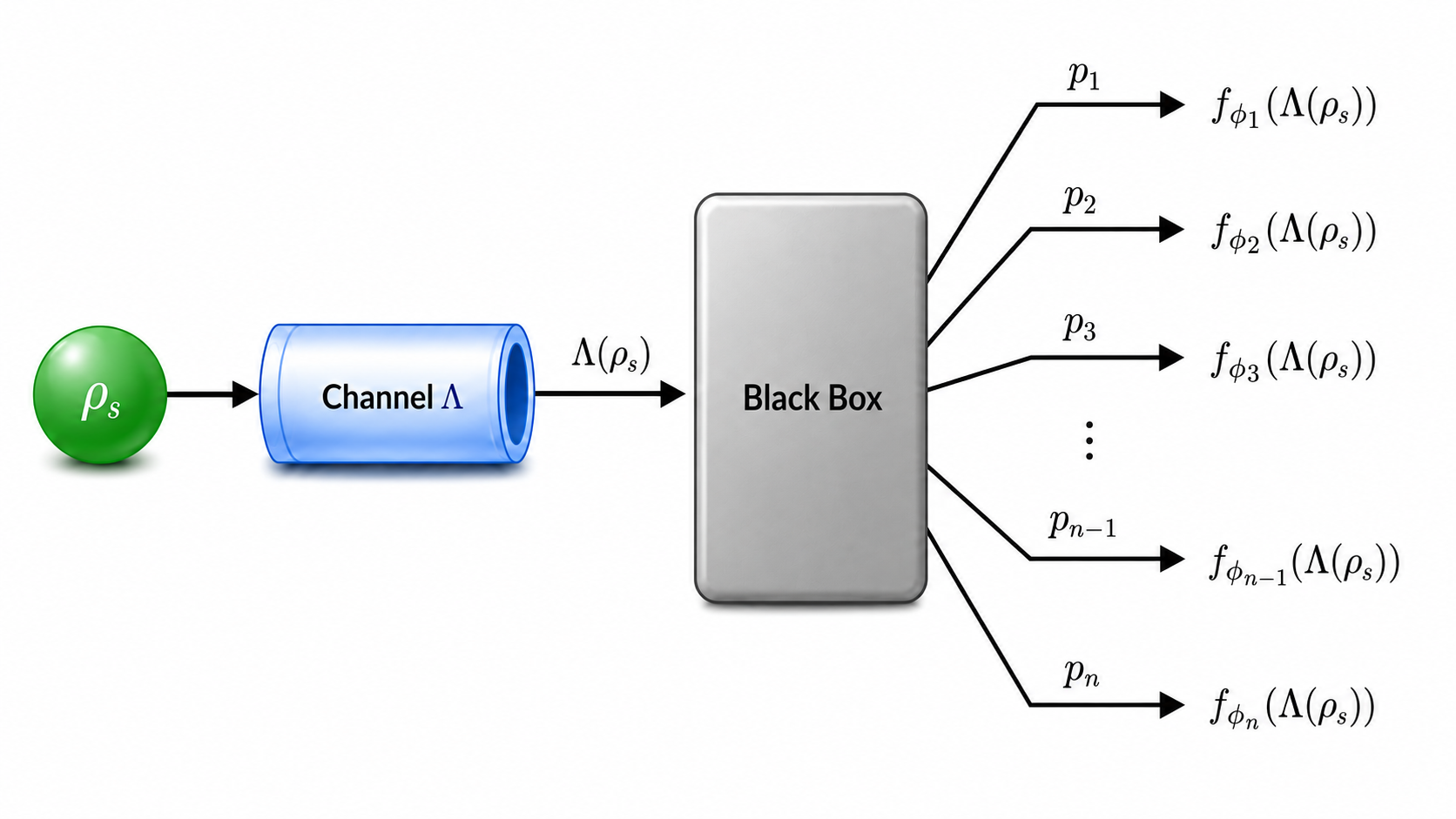}
\caption{Schematic diagram for phase discrimination}
\label{fig:phase discrimination}
\end{figure}
Let for a $\Lambda(\rho_s)$,  we can find an incoherent unitary operator $U$ which maps $\Lambda(\rho_s)$ to $\tilde{\sigma}=U^{\dagger}\Lambda(\rho_s) U$ with entries $\tilde{\sigma}_{ij}=|(\Lambda(\rho_s))_{ij}|$ then using \eqref{eq:fraction with robustness} we get 
\begin{equation}
    \max_{\Theta}\frac{p_{\Theta}^{\text{succ}}(\Lambda(\rho_s))}{p_{\Theta}^{\text{succ}}(\delta_s)}=1+C_{R}(\Lambda(\rho_s))=d C_{F}(\Lambda(\rho_s)).
    \label{eq:advantage ratio phase discrimination}
\end{equation}
Now if we use $\Lambda(\rho_s')$ (see equation \eqref{eq:final state} for $\rho_s'$) instead of $\Lambda(\rho_s)$ as an input to the phase discrimination task one of the phases $\{\phi_k|1,2,...,n\}$ with a prior probability $p_k$ has been encoded upon $\Lambda(\rho_s')$. The \textit{maximum advantage ratio} by using $\Lambda(\rho_s')$ in all possible phase discrimination game $\Theta$ is given by 
\begin{equation}
    \max_{\Theta}\frac{p_{\Theta}^{\text{succ}}(\Lambda(\rho_s'))}{p_{\Theta}^{\text{succ}}(\delta_s)}=1+ C_{R}(\Lambda(\rho_s')).
    \label{eq:channel advantage ratio}
\end{equation}
Using equation \eqref{eq:fraction with robustness inequality} in equation \eqref{eq:channel advantage ratio} we have
\begin{equation}
   \lim_{n\to \infty}\max_{\mathcal{E}\in \text{IO}} \left(\max_{\Theta}\frac{p_{\Theta}^{\text{succ}}(\Lambda(\rho_s'))}{p_{\Theta}^{\text{succ}}(\delta_s)}\right)\geq \lim_{n\to \infty}\max_{\mathcal{E}\in \text{IO}} dC_F(\Lambda(\rho_s')).
\end{equation}
Assuming the quantum channel $\Lambda$ satisfy conditions of either theorem \ref{th:main 2nd theorem} or \ref{th:1st main theorem}, we have from equation \eqref{eq:main inequality} or \eqref{eq:main inequality 3}
\begin{figure}[ht]
\centering
\includegraphics[width=0.45\textwidth]{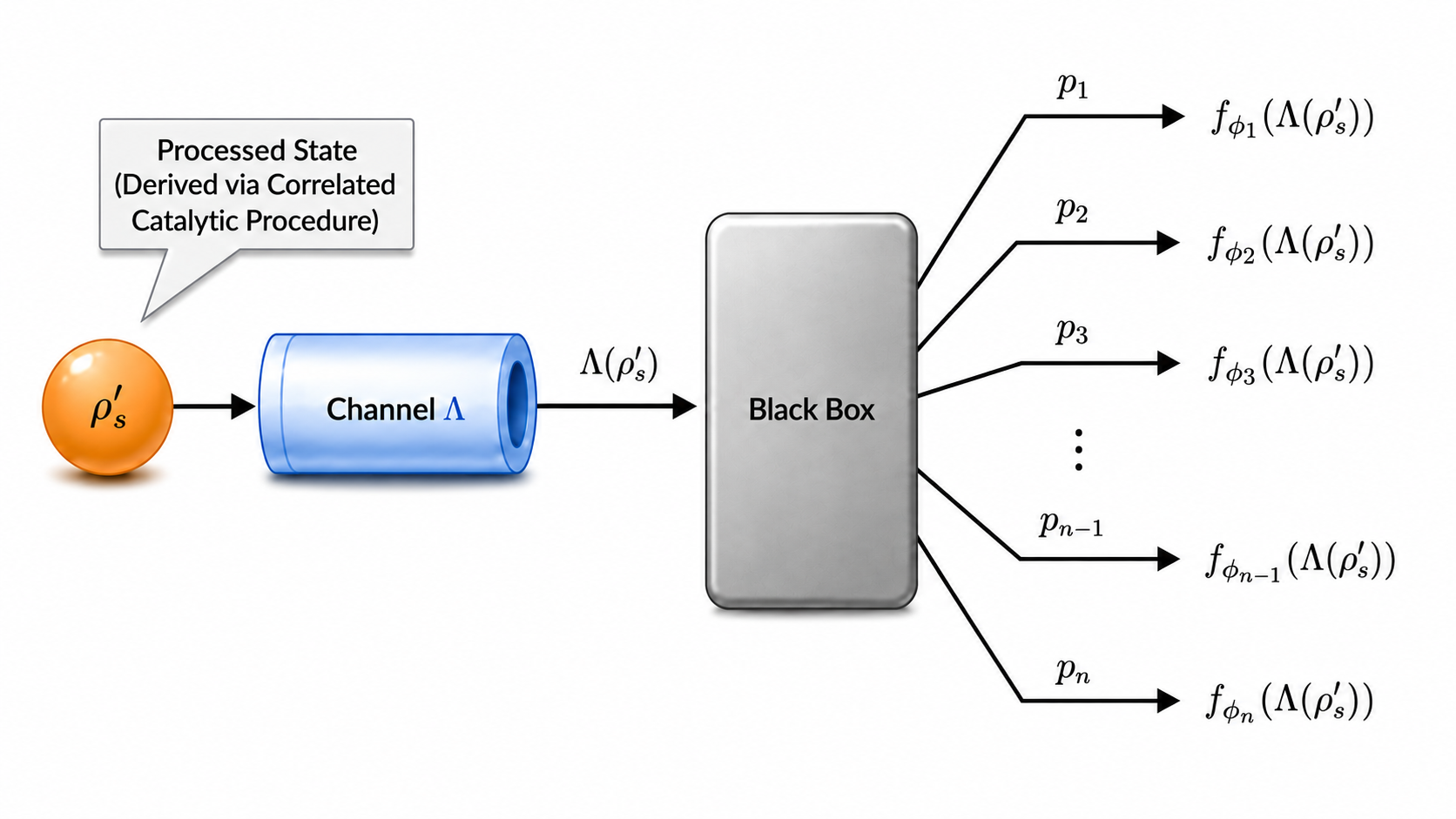}
\caption{Schematic diagram for catalytic phase discrimination}
\label{fig:phase discrimination catalytic}
\end{figure}
\begin{align*}
     \lim_{n\to \infty}\max_{\mathcal{E}\in \text{IO}} \left(\max_{\Theta}\frac{p_{\Theta}^{\text{succ}}(\Lambda(\rho_s'))}{p_{\Theta}^{\text{succ}}(\delta_s)}\right)&\geq \lim_{n\to \infty}\max_{\mathcal{E}\in \text{IO}} dC_F(\Lambda(\rho_s')) \\ \geq dC_F(\Lambda(\rho_s))&=\max_{\Theta}\frac{p_{\Theta}^{\text{succ}}(\Lambda(\rho_s))}{p_{\Theta}^{\text{succ}}(\delta_s)}.
\end{align*}
This demonstrates that catalysis can protect or amplify coherence through a channel, leading to better post-channel performance in an operational task like phase discrimination. 

\section{Example 1: Quantum Addition Channel}
\label{sec:Example}
In this section, we will provide an example of a quantum channel that commutes with a particular type of SIO. First, we will prove a theorem regarding the multiplicative property of this particular type of SIO.
\begin{theorem}
\label{th:Dipayan's 2nd theorem}
    If a CPTP map $\mathcal{E}$ is of the  form 
   $\mathcal{E}(X)=A \odot X$,
    with $A \succeq 0$ and $A_{ii}=1$ for all $i$ ( so $\mathcal{E}$ is SIO by theorem \ref{th:Dipayan's 1st theorem}), then for any incoherent state $\sigma$ and any state $\rho$, the following multiplicative relation holds, 
    $$\mathcal{E}(\rho)\mathcal{E}(\sigma)=\mathcal{E}(\rho \sigma).$$
    
\end{theorem}

\begin{proof}
    let $\mathcal{E}$ be of the form
    $\mathcal{E}(X)=A \odot X$,
    with $A \succeq 0$ and $A_{ii}=1$ for all $i$.
    We will show for every density operator $\rho$ and every incoherent state $\sigma$
    $$\mathcal{E}(\rho \sigma)=\mathcal{E}(\rho) \mathcal{E}(\sigma).$$
    Let $\sigma=diag(\sigma_{11}, \sigma_{22},.....,\sigma_{dd})$. Since $\sigma$ is diagonal,
    $$(\rho \sigma)_{ij}=\sum_k \rho_{ik} \sigma_{kj}=\rho_{ij}\sigma_{jj}.$$
    Now, for all $i,j$ we have
    \begin{equation}
    \label{eq:Dipayan's l.h.s.}
        (\mathcal{E}(\rho \sigma))_{ij}=(A \odot (\rho \sigma))_{ij}=A_{ij}(\rho \sigma)_{ij}=A_{ij}\rho_{ij}\sigma_{jj}.
    \end{equation}
    Again using $A_{ii}=1$ we can prove that,
    $$\mathcal{E}(\sigma)= A \odot \sigma= diag(A_{11}\sigma_{11}, A_{22} \sigma_{22},...., A_{dd}\sigma_{dd})=\sigma.$$ 
    Therefore for all $i,j$ using $\mathcal{E}(\sigma)=\sigma$, we have
    \begin{equation}
    \label{eq:Dipayan's r.h.s}
        ( \mathcal{E}(\rho) \mathcal{E}(\sigma))_{ij}=( (A \odot \rho) \sigma)_{ij}=A_{ij} \rho_{ij} \sigma_{jj}.
    \end{equation}
    Comparing equation \eqref{eq:Dipayan's l.h.s.} and \eqref{eq:Dipayan's r.h.s}, we get $$\mathcal{E}( \rho \sigma)=\mathcal{E}(\rho)\mathcal{E}(\sigma) \text{ or equivalently }$$ 
    \begin{equation}
        \label{eq:Dipayn's theorem equation}
        A\odot(\rho\sigma)=(A\odot\rho)(A\odot\sigma)
    \end{equation}
\end{proof}
\textbf{Quantum Addition Channel:} Now, let us define the quantum addition channel \cite{mukhopadhyay2020quantum} $\Lambda$ by 
\begin{multline}
\label{eq:quantum addition channel}
    \Lambda(\mathcal{E}(\rho_A))=\mathcal{E}(\rho_A) \boxplus_\alpha \mathcal{E}(\sigma_B)\\ =\alpha \mathcal{E}(\rho_A) + (1-\alpha) \mathcal{E}(\sigma_B)-i\sqrt{\alpha(1-\alpha)}[\mathcal{E}(\rho_A),\mathcal{E}(\sigma_B)],
\end{multline}
where $\mathcal{E}$ is any IO. The above channel $\Lambda$ can also be represented as 
$$\Lambda(\mathcal{E}(\rho_A))=Tr_B[U_\alpha (\mathcal{E}(\rho_A) \otimes \mathcal{E}(\sigma_B)) U_\alpha^\dagger]$$
Where $$U_\alpha=\sqrt{\alpha}I+i\sqrt{1-\alpha}S$$
Where $S$ is swap operator and $S\left(\ket{i}_A \otimes \ket{j}_B \right)=\ket{j}_A \otimes \ket{i}_B$, $S=\sum_{i,j} \ket{i}\bra{j}_A \otimes \ket{j} \bra{i}_B$\\\\
We next prove that the quantum addition channel $\Lambda$ commutes with the special type of SIO $\mathcal{E}$ defined as  $\mathcal{E}(X)=A \odot X$,
    with $A \succeq 0$ and $A_{ii}=1$ for all $i$.
\begin{theorem}
    A quantum addition channel $\Lambda$ defined by \eqref{eq:quantum addition channel} always commutes with an SIO $\mathcal{E}$ where $\mathcal{E}$ is defined as  $\mathcal{E}(X)=A \odot X$, with $A \succeq 0$ and $A_{ii}=1$ for all $i$. 
\end{theorem}

\begin{proof}
Since SIO $\mathcal{E}$ is defined as  $\mathcal{E}(X)=A \odot X$, with $A \succeq 0$ and $A_{ii}=1$ for all $i$ then by theorem \ref{th:Dipayan's 2nd theorem} we have that $\mathcal{E}(\rho)\mathcal{E}(\sigma)=\mathcal{E}(\rho\sigma).$ Now as $\sigma$ is an incoherent state and $\mathcal{E}$ is an SIO, it is obvious that $\mathcal{E}(\sigma)$ is incoherent.
    \small{
    \begin{align*}
        & \Lambda(\mathcal{E}(\rho))\\&=\mathcal{E}(\rho) \boxplus_\alpha \mathcal{E}(\sigma)\\ &=\alpha \mathcal{E}(\rho) + (1-\alpha) \mathcal{E}(\sigma)-i\sqrt{\alpha(1-\alpha)}[\mathcal{E}(\rho),\mathcal{E}(\sigma)]\\
        &=\alpha \mathcal{E}(\rho) + (1-\alpha) \mathcal{E}(\sigma)-i\sqrt{\alpha(1-\alpha)}\left(\mathcal{E}(\rho)\mathcal{E}(\sigma)-\mathcal{E}(\sigma)\mathcal{E}(\rho)\right)\\
        &=\alpha \mathcal{E}(\rho) + (1-\alpha) \mathcal{E}(\sigma)-i\sqrt{\alpha(1-\alpha)}\left(\mathcal{E}(\rho\sigma)-\mathcal{E}(\sigma\rho)\right)\\
        &=\alpha \mathcal{E}(\rho) + (1-\alpha) \mathcal{E}(\sigma)-i\sqrt{\alpha(1-\alpha)}\left(\mathcal{E}(\rho\sigma-\sigma\rho)\right)\\
        &=\mathcal{E}(\alpha\rho + (1-\alpha) \sigma-i\sqrt{\alpha(1-\alpha)}\left(\rho\sigma-\sigma\rho\right))\\
        &=\mathcal{E}(\rho\boxplus_\alpha\sigma)\\
        &=\mathcal{E}(\Lambda(\rho))
    \end{align*}
    }
\end{proof}
Thus, we can conclude that the result of theorem \ref{th:main 2nd theorem} for a quantum addition channel $\Lambda$ is valid.

\section{Example 2: Qutrit Phase Damping Channel}
\label{sec:example 1}
In this section, we present a specific example of a qutrit phase-damping channel that illustrates Theorem \ref{th:1st main theorem}.

Let us consider a single-party $3$ dimensional state $\ket{\psi}_{A_1}=\frac{\ket{0}+\ket{1}}{\sqrt{2}}$ and a biparty $3$ dimensional state $\ket{\phi}_{A_{1}A_{2}}=\sqrt{\frac{z}{3}}\ket{00}+\sqrt{\frac{z}{3}}\ket{01}+\sqrt{\frac{z}{3}}\ket{02}+\sqrt{1-z}\ket{12}$.
$$\text{Diag}(\ket{\psi}\bra{\psi}_{A_{1}}^{\otimes 2})=\left(\frac{1}{4},\frac{1}{4},\frac{1}{4},\frac{1}{4},0,0,0,0,0\right),$$ $$\text{ } \text{Diag}(\ket{\phi}\bra{\phi}_{A_{1}A_{2}})=\left(\frac{z}{3},\frac{z}{3},\frac{z}{3},1-z,0,0,0,0,0\right).$$

For $z\geq \frac{3}{4}$, $\text{Diag}(\ket{\phi}\bra{\phi}_{A_{1}A_{2}})$ majorizes $\text{Diag}(\ket{\psi}_{A_{1}}^{\otimes 2})$, and $\ket{\psi}\bra{\psi}_{A_{1}}^{\otimes 2}$ can be transformed to $\ket{\phi}_{A_{1}A_{2}}$. \\

The catalysis state when $n=2$ is
$$\tau_{A_2R}=\frac{1}{2}(\;\zeta_{A_2}\otimes\ket{1}\bra{1}_R+\ket{\psi}\bra{\psi}_{A_2}\otimes\ket{2}\bra{2}_R\;),\text{ where}$$ 
\begin{equation*}
\begin{aligned}
\zeta_{A_2}&=\Tr_{A_1}(\ket{\phi}\bra{\phi}_{A_1A_2})\\&=\frac{z}{3}(\;\ket{0}\bra{0}+\ket{1}\bra{1}+\ket{2}\bra{2}+\ket{0}\bra{1}+\ket{1}\bra{0}\\&+\ket{0}\bra{2}+\ket{2}\bra{0}+\ket{1}\bra{2}+\ket{2}\bra{1}\;)+(1-z)\ket{2}\bra{2}.
\end{aligned}
\end{equation*}
\begin{equation*}
\begin{aligned}
\text{Then, }\ket{\psi}\bra{\psi}_{A_1}\otimes\tau_{A_2 R}&=\frac{1}{2}(\ket{\psi}\bra{\psi}_{A_1}\otimes \zeta_{A_2}\otimes\ket{1}\bra{1}_R\\&+\ket{\psi}\bra{\psi}_{A_1}\otimes\ket{\psi}\bra{\psi}_{A_2}\otimes\ket{2}\bra{2}_R).
\end{aligned}
\end{equation*}
By following step (1)-(3) of theorem \ref{th:catalytic IO} we can transfer  $\ket{\psi}\bra{\psi}_{A_1}\otimes\tau_{A_2 R}$, to $\beta_{A_1A_2R},$
\begin{equation*}
\begin{aligned}
    \beta_{A_1A_2R}=\frac{1}{2}(\;\ket{\phi}\bra{\phi}_{A_1A_2}&\otimes\ket{1}\bra{1}_R\\&+ \zeta_{A_1}\otimes\ket{\psi}\bra{\psi}_{A_2}\otimes\ket{2}\bra{2}_R\;)
\end{aligned}
\end{equation*}
It is important to point out that when we trace out $A_1$ from $\beta_{A_1A_2R}$, the resulting state is $\Tr_{A_1}\beta_{A_1A_2R}=\tau_{A_2R}$. Conversely, tracing out $A_2R$ yields $\rho'_{A_1}=\Tr_{A_2R}(\beta_{A_1A_2R})=\frac{1}{2}(\zeta'_{A_1}+\zeta_{A_1}), \text{ where } \zeta'_{A_1}=\Tr_{A_2}\ket{\phi}\bra{\phi}_{A_1A_2}.$
\begin{equation*}
\begin{aligned}
\rho'_{A_1}=\frac{2z}{3}\ket{0}\bra{0}&+\frac{3-2z}{6}(\; \ket{1}\bra{1}+\ket{2}\bra{2}\;)\\&+\left(\frac{z}{6}+\frac{\sqrt{z(1-z)}}{2\sqrt{3}}\right)(\;\ket{0}\bra{1}+\ket{1}\bra{0}\;)\\
&+\frac{z}{6}(\;\ket{1}\bra{2}+\ket{2}\bra{1}+
\ket{0}\bra{2}+ \ket{2}\bra{0}\;).\\
\end{aligned}
\end{equation*}
Now let us consider a qutrit phase damping (dephasing) channel $\Lambda_{\text{deph}}$ whose action on the input state $\rho$ gives
\begin{equation}
\label{eq:dephasing channel}
    \Lambda_{\text{deph}}(\rho)=(1-p)\rho+p\sum_{i=0}^{2}\ket{i}\bra{i}\rho\ket{i}\bra{i}
\end{equation}
\begin{figure}[!htbp] % use * to span both columns
    \hspace{-10cm}
    \begin{minipage}{0.48\textwidth}
    \vspace{1cm}
       \includegraphics[width=\linewidth]{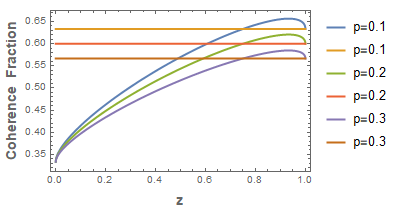}
       \vspace{-0.5cm}
         \caption{The horizontal lines represents $\frac{2-p}{3}$ and the curves represents $\frac{1}{3}+\frac{(1-p)z}{3}+\frac{(1-p)\sqrt{z(1-z)}}{3\sqrt{3}}$ for respective values of $p$.}
         \label{fig:1st fig}
         %\vspace{0.5cm}
           \includegraphics[width=\linewidth]{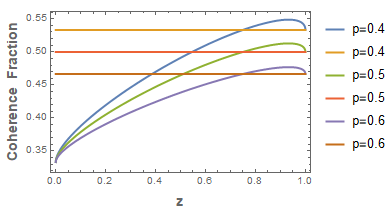}
           \vspace{-0.5cm}
      \caption{The horizontal lines represents $\frac{2-p}{3}$ and the curves represents $\frac{1}{3}+\frac{(1-p)z}{3}+\frac{(1-p)\sqrt{z(1-z)}}{3\sqrt{3}}$ for respective values of $p$.}
      \label{fig:2nd fig}
      % \vspace{1cm}
        \includegraphics[width=\linewidth]{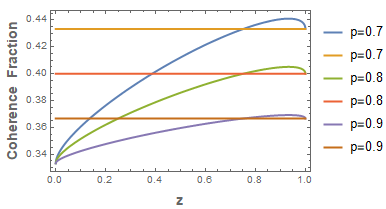}
        \vspace{-0.5cm}
        \caption{The horizontal lines represents $\frac{2-p}{3}$ and the curves represents $\frac{1}{3}+\frac{(1-p)z}{3}+\frac{(1-p)\sqrt{z(1-z)}}{3\sqrt{3}}$ for respective values of $p$.}
         \label{fig:3rd fig}
    \end{minipage}
   \hspace{-10cm} 
\end{figure}	
The Kraus operators of the dephasing channel $\Lambda_{\text{deph}}$ are $K_0=\sqrt{1-p}I_3,\; K_1=\sqrt{p}\ket{0}\bra{0},\; K_2=\sqrt{p}\ket{1}\bra{1},\; K_3=\sqrt{p}\ket{2}\bra{2}.$
Action of the dephasing channel $\Lambda_{\text{deph}}$ on $\ket{\psi}\bra{\psi}_{A_1}$ gives the following output state
\begin{multline}
     \label{eq:1st output state}
    \Lambda_{\text{deph}}(\ket{\psi}\bra{\psi}_{A_1})=\frac{1}{2}(\ket{0}\bra{0}+\ket{1}\bra{1})\\+\frac{1-p}{2}(\ket{0}\bra{1}+\ket{1}\bra{0}).
\end{multline}
Action of the dephasing channel $\Lambda_{\text{deph}}$ on $\rho_{A_1}'$ give the following output state
\small{
\begin{equation}
    \label{eq:2nd output state}
    \begin{aligned}
\Lambda_{\text{deph}}(\rho'_{A_1})&=\frac{2z}{3}\ket{0}\bra{0}+\frac{3-2z}{6}(\; \ket{1}\bra{1}+\ket{2}\bra{2}\;)\\&+\left(\frac{z}{6}+\frac{\sqrt{z(1-z)}}{2\sqrt{3}}\right)(1-p)(\;\ket{0}\bra{1}+\ket{1}\bra{0}\;)\\
&+\frac{z(1-p)}{6}(\;\ket{1}\bra{2}+\ket{2}\bra{1}+
\ket{0}\bra{2}+ \ket{2}\bra{0}\;).\\
\end{aligned}
\end{equation}
}
Coherence fraction of the output states \eqref{eq:1st output state} and \eqref{eq:2nd output state} are respectively 
\begin{equation}
    \label{eq:1st output C.F.}
    C_F(\Lambda_{\text{deph}}(\ket{\psi}\bra{\psi}_{A_1}))=\frac{2-p}{3},
\end{equation}
\begin{equation}
    \label{eq:2nd output C.F.}
    C_F(\Lambda_{\text{deph}}(\rho'_{A_1}))=\frac{1}{3}+\frac{(1-p)z}{3}+\frac{(1-p)\sqrt{z(1-z)}}{3\sqrt{3}}.
\end{equation}
We can see that if $p=1$ then the above channel is reduced to the completely dephasing channel $(\Delta)$, and from equation \eqref{eq:1st output C.F.} and \eqref{eq:2nd output C.F.}  $ C_F(\Lambda_{\text{deph}}(\ket{\psi}\bra{\psi}_{A_1}))=C_F(\Lambda_{\text{deph}}(\rho'_{A_1}))=\frac{1}{3},$ and $C_F(\ket{\psi}\bra{\psi}_{A_1})=\frac{2}{3}.$ But for the other cases of $p\in(0,1)$, we can see from figure \ref{fig:1st fig}, \ref{fig:2nd fig} and \ref{fig:3rd fig} a clear improvement of coherence fraction by using the processed state $\rho_{A_1}'$ instead of $\ket{\psi}\bra{\psi}_{A_1}$.

\section{SIO in Schur-Multiplier Form}
\label{sec:SIO}
In this section, we have started with a necessary and sufficient condition for an incoherent state-preserving CPTP map to be a special kind of SIO. After that, we have also provided numerical evidence to show that a particular SIO map, which is not of the Schur multiplier form, does not satisfy the condition of the theorem.
 \begin{theorem}
 \label{th:Dipayan's 3rd theorem}
     For the fixed incoherent basis $\left\{ \ket{i} \right\}_{i=1}^{d}$, let $\sigma=\sum_{i=1}^{d} p_i \ket{i} \bra{i}$ be a non-degenerate incoherent state. Let $\mathcal{E}$ be a CPTP map with $\mathcal{E(I)} \subseteq \mathcal{I}$ then the condition $\mathcal{E}(\rho)\mathcal{E}(\sigma)=\mathcal{E}(\rho \sigma)$ for any state $\rho$, is the necessary and sufficient condition for the CPTP map to be a SIO of the form $$\mathcal{E}(\rho)=U(A\odot \rho)U^{\dagger},$$ where $A\succeq0$ with $A_{ii}=1$ for all $i$ and $U$ is an incoherent unitary operation.
   \end{theorem}
   \begin{proof}
   Let the condition $\mathcal{E}(\rho)\mathcal{E}(\sigma)=\mathcal{E}(\rho \sigma)$ holds. \\\\
       Let $\mathbb{D}$ denote the set of diagonal matrices in the incoherent basis and
       \begin{equation}
       \label{eq:diagonal notation}
           \mathbb{D}=\left\{ \sum_i x_{i} \ket{i} \bra{i}: x_i \in \mathbb{R} \right\} \\
       \end{equation}
       Because the diagonal entries $p_i$ are distinct, $\sigma$ has non degenerate spectrum on $\mathbb{D}$.\\\\
       %By Lagrange Interpolation, for each diagonal projector $P_i=\ket{i}\bra{i}$ there exist a polynomial $f_i$ such that 
      % \begin{equation}
         %  P_i=f_i(\sigma)
       %\end{equation}
       Let
       \begin{equation}
       \label{eq: Incoherent state notation}
        \sigma=\sum_{i=1}^{d} p_i \ket{i} \bra{i}
       \end{equation}
       with $p_k \neq p_{l}$ for $k \neq l$ and $p_k >0$ \\\\
  We Know that for any polynomial $f(x)$,
  $$f(\sigma)=\sum_{i=1}^d f(p_i) \ket{i}\bra{i}$$
       %Now $\sigma=\sum_{k=1}^{d} p_k \ket{k} \bra{k}$\\\\
       %This gives $\sigma^2=\sum_{k=1}^{d} p_k^2 \ket{k} \bra{k}$
      % and $\sigma^n=\sum_{k=1}^{d} p_k^n \ket{k} \bra{k}$\\\\
       %Now 
       %\begin{align*}
          % f(\sigma)&=a_0 I+ a_1 \sigma+ a_2 \sigma^2+....+a_{m-1} \sigma^{m-1}\\
          % &=\sum_k (a_0+a_1 p_k + a_2 p_k^2+....+a_{m-1} p_k^{m-1}) \ket{k}\bra{k}\\
          % &=\sum_k f(p_k) \ket{k} \bra{k}
      % \end{align*}
       Let us define a polynomial $f(t)$ such that 
       \begin{equation*}
           f(t)=\sum_{i=1}^d x_i f_i(t)
       \end{equation*}
         %  $$f_i(x)=\frac{(x-p_{1})(x-p_{2})...(x-p_{i-1})(x-p_{i+1})...(x-p_{d})}{(p_i-p_1)(p_i-p_2)...(p_i-p_{i-1})(p_i-p_{i+1})...(p_i-p_d)}$$
          where, 
           $$f_i(t)=\prod_{k=1, k \neq i}^d \frac{(t-p_k)}{(p_i-p_k)}$$
       This gives for $t=p_i,\text{ } f_i(p_i)=1$ and for $t=p_j$ with $ j \neq i $ $f_i(p_j)=0$. 
       %\begin{align*}
          % f_i(\sigma)&=\sum_{k=1}^d f_i(p_k) \ket{k} \bra{k}\\
           % &= \ket{i} \bra{i}\\
          % &= P_i
       %\end{align*}
      This clearly gives $f(p_i)=x_i$
       %\begin{equation}
          % P_i=f_i (\sigma) \tag{24b}
      % \end{equation}
      Thus 
      $$f(\sigma)=\sum_i f(p_i) \ket{i} \bra{i}=\sum_i x_i \ket{i} \bra{i}$$
       Therefore for every $D \in \mathbb{D}$ is a polynomial in $\sigma$ i.e.
       \begin{equation}
       \label{eq:25}
           D=f(\sigma) 
       \end{equation}
       We have assumed $\mathcal{E}(\rho)\mathcal{E}(\sigma)=\mathcal{E}(\rho \sigma)$.\\\\
      % We replace $\rho$ by $\rho \sigma^{n}$ and  we get for $n \geq 1$
      % Next let us prove that for any $n \geq 1$
       %\begin{equation}
       %    \mathcal{E}(\rho) \mathcal{E}(\sigma^n)=\mathcal{E}(\rho \sigma^n)
       %\end{equation}
      % Now for $n=1$, $\mathcal{E}(\rho) \mathcal{E}(\sigma)=\mathcal{E}(\rho \sigma)$ which is already given.\\\\
      % Let for $n=m$, the condition $\mathcal{E}(\rho)\mathcal{E}(\sigma^m)=\mathcal{E}(\rho \sigma^m)$ holds.
      % Now
      % \begin{align}
      % \mathcal{E}(\rho \sigma^m) \mathcal{E}(\sigma) &=\mathcal{E}(\rho \sigma^m \sigma) \tag{26a}\\
      % &=\mathcal{E}(\rho \sigma^{m+1}) \tag{26b}
      % \end{align}
      % Again
       %\begin{equation}
         %  \mathcal{E}(\rho \sigma^m)=\mathcal{E}(\rho) \mathcal{E}(\sigma^m) \tag{26c}
      % \end{equation}
      % From $(26b)$ and $(26c)$ we have
      % \begin{align*}
         %  \left( \mathcal{E}(\rho) \mathcal{E}(\sigma^m) \right) \mathcal{E}(\sigma)&=\mathcal{E}(\rho \sigma^{m+1})\\
         %  \textbf{or} \;\;\mathcal{E}(\rho) \left( \mathcal{E}(\sigma^m)  \mathcal{E}(\sigma)\;\; \right)&=\mathcal{E}(\rho \sigma^{m+1})\\
         %  \textbf{or}\;\; \mathcal{E}(\rho) \mathcal{E}(\sigma^{m+1})&=\mathcal{E}(\rho \sigma^{m+1})
       %\end{align*}
        By principle of mathematical induction we have $\mathcal{E}(\rho) \mathcal{E}(\sigma^n)=\mathcal{E}(\rho \sigma^n)$ holds for all $n \geq 1$.
       Now 
       \begin{align*}
       \mathcal{E}(\rho) \mathcal{E}(f(\sigma))&=\mathcal{E}(\rho) \mathcal{E}(a_0 I+\sum_{n=1}^{d-1} a_n \sigma^n) \\
       &=\mathcal{E}(\rho) (a_0 \mathcal{E}(I)+\sum_{n=1}^{d-1} a_n \mathcal{E}( \sigma^n)) \\
       &=a_0 \mathcal{E}(\rho) \mathcal{E}(I)+\sum_{n=1}^{d-1} a_n \mathcal{E}(\rho) \mathcal{E}(\sigma^n)\\
       &=a_0 \mathcal{E}(\rho)+\sum_{n=1}^{d-1} a_n \mathcal{E}(\rho) \mathcal{E}(\sigma^n)\\
       &=a_0 \mathcal{E}(\rho)+\sum_{n=1}^{d-1}a_n \mathcal{E}(\rho \sigma^n)\\
       &=\mathcal{E}(a_0 \rho + \sum_{n=1}^{d-1}a_n \rho \sigma^n )\\
       &=\mathcal{E}(\rho(a_0 I+ \sum_{n=1}^{d-1} a_n \sigma^n))\\
       &=\mathcal{E}(\rho f(\sigma))
       \end{align*}
       Thus we  have 
       \begin{equation}
       \label{eq:27}
           \mathcal{E}(\rho) \mathcal{E}(f(\sigma))=\mathcal{E}(\rho f(\sigma))
       \end{equation}
       
       From \eqref{eq:25} and \eqref{eq:27} we can write 
       \begin{equation}
       \label{eq:28}
           \mathcal{E}(\rho)\mathcal{E}(D)=\mathcal{E}(\rho D)
       \end{equation}

       We have assumed $\mathcal{E}(\rho)\mathcal{E}(\sigma)=\mathcal{E}(\rho \sigma)$, taking conjugate transpose we get,
       \begin{align*}
           (\mathcal{E}(\rho) \mathcal{E}(\sigma))^\dagger=(\mathcal{E}(\rho \sigma))^\dagger
       \end{align*}   
       The left-hand side of the above expression equals to 
          $$(\mathcal{E}(\rho) \mathcal{E}(\sigma))^\dagger=(\mathcal{E}(\sigma))^\dagger (\mathcal{E}(\rho))^\dagger$$
           Since $\mathcal{E}$ is a CPTP map, it preserves adjoints, so we have
          $$ \mathcal{E}(X^\dagger)=(\mathcal{E}(X))^\dagger$$
          This will imply $$(\mathcal{E}(\rho) \mathcal{E}(\sigma))^\dagger=\mathcal{E}(\sigma^\dagger)\mathcal{E}(\rho^\dagger)=\mathcal{E}(\sigma)\mathcal{E}(\rho)$$
          Now, Right Hand Side of the expression gives 
          $$(\mathcal{E}(\rho \sigma))^\dagger=\mathcal{E}((\rho \sigma)^\dagger )=\mathcal{E}(\sigma^\dagger \rho^\dagger)=\mathcal{E}(\sigma \rho)$$
          Thus we have 
\begin{equation}
    \mathcal{E}(\sigma)\mathcal{E}(\rho)=\mathcal{E}(\sigma \rho)
\end{equation}
The same arguments will imply 
\begin{equation}
\label{eq:31}
    \mathcal{E}(D) \mathcal{E}(\rho)=\mathcal{E}(D \rho)\;\;\;\text{for all $\rho$, for all $D \in \mathbb{D}$}
\end{equation}
Let $P_i=\ket{i}\bra{i}$. Taking $\rho=P_i$ and $D=P_j$ and substituting in \eqref{eq:28} we get $\mathcal{E}(P_i)\mathcal{E}(P_j)=\mathcal{E}(P_i P_j)$.\\\\
Now $\mathcal{E}(P_i P_j)=\delta_{ij}\mathcal{E}(\ket{i}\bra{j})$\\\\
This will imply $\mathcal{E}(P_i)\mathcal{E}(P_j)=0$ for $i \neq j$ and $\mathcal{E}(P_i)\mathcal{E}(P_i)=\mathcal{E}(P_i)$ for $i=j$.\\\\
So, $\mathcal{E}(P_i)$ are pairwise orthogonal idempotents. Since $\mathcal{E}(\mathcal{I}) \subseteq \mathcal{I}$ and each $P_i$ is incoherent, therefore for each $i$ there exists a permutation $\pi$ such that $\mathcal{E}(P_i)=P_{\pi(i)}$. \\

Let $E_{ij}=\ket{i}\bra{j}$. Now we use \eqref{eq:28} with $D=P_k$
\begin{equation}
\mathcal{E}(E_{ij})\mathcal{E}(P_k)=\mathcal{E}(E_{ij}P_k)=\delta_{jk}\mathcal{E}(E_{ik})
\end{equation}

Similarly, We use \eqref{eq:31} with $D=P_k$
\begin{equation}
\mathcal{E}(P_k)\mathcal{E}(E_{ij})=\mathcal{E}(P_kE_{ij})=\delta_{ki}\mathcal{E}(E_{kj})
\end{equation}

Thus the only possibility is that 
\begin{equation}
    \mathcal{E}(E_{ij})=c_{ij}\ket{\pi(i)}\bra{\pi(j)}
\end{equation}
for some scalars $c_{ij} \in \mathbb{C}$. It is obvious that $c_{ii}=1$ because $\mathcal{E}(P_i)=\ket{\pi(i)}\bra{\pi(j)}$\\\\
A CPTP map whose action on matrix units has the form $(37)$ is precisely a Schur multiplier channel in that basis
\begin{equation}
    \mathcal{E}(\rho)=U (M \odot \rho) U^\dagger
\end{equation}
Where $U\ket{i}=\ket{\pi(i)}$ and $M=[c_{ij}]$ is a positive semidefinite matrix with $c_{ii}=1$. Clearly, this unitary $U$ is an incoherent unitary, since it is just a permutation of the incoherent basis. Hence $\mathcal{E}$ is an SIO.\\
%Such maps admit Kraus operators of the form $K_{\alpha}=U D_{\alpha}$ with each $D_\alpha$ diagonal in the incoherent basis. Thus $D_\alpha$ has at most one non-zero entry in each column and each row. Further Unitary $U$ that maps basis projectors into basis projectors  must be an incoherent unitary.\\\\

Conversely, let $\mathcal{E}(\rho)=U(A\odot \rho)U^{\dagger},$ where $U$ is incoherent unitary and $A\succeq 0$ with $A_{ii}=1$ for all $i$. Now, for any state $\rho$ and a non-degenerate incoherent state $\sigma$ we have,
\begin{align*}
    \mathcal{E}(\rho)\mathcal{E}(\sigma)&=U(A\odot \rho)U^{\dagger}U(A\odot \rho)U^{\dagger}\\
    &=U(A\odot \rho)I(A\odot \sigma)U^{\dagger}\\
    &=U(A\odot \rho)(A\odot \sigma)U^{\dagger}\\
    &=U(A\odot (\rho\sigma))U^{\dagger} \quad [\text{Using equation \eqref{eq:Dipayn's theorem equation}}]\\
    &=\mathcal{E}(\rho\sigma).
\end{align*}
Hence, $\mathcal{E}(\rho)\mathcal{E}(\sigma)=\mathcal{E}(\rho\sigma).$ 
\end{proof}

Now, we perform a numerical study of the identity $$\mathcal{E}(\rho)\mathcal{E}(\sigma)=\mathcal{E}(\rho \sigma),$$
for a non-degenerate incoherent state $\sigma$, using a class of strictly incoherent operations $\mathcal{E}$ that are not of the form 
\begin{equation}
\label{eq 51}
    \mathcal{E}(\rho)= U ( C \odot \rho) U^\dagger
\end{equation}
We consider the qubit incoherent state 
$\sigma=\begin{pmatrix}
    0.3 & 0 \\
    0 & 0.7
\end{pmatrix}$,
which is non-degenerate. Let us define $\mathcal{E}$ as
$$\mathcal{E}(\rho)=\sum_{n=1}^2 K_n \rho K_n^{\dagger},$$ 
with Kraus operators 
$K_1=\begin{pmatrix}
    1 & 0 \\
    0 & \sqrt{0.6}
\end{pmatrix},\;K_2=\begin{pmatrix}
0 & \sqrt{0.4}\\
0 & 0
\end{pmatrix}$. These operators satisfy the completeness condition $\sum_n K_n^\dagger K_n=I$. From our construction, it is clear that $\mathcal{E}$ is SIO but not in the form of \ref{eq 51}.

Now we generate an ensemble of random density matrices $$\rho_k=\frac{G_k G_k^\dagger}{Tr(G_k G_k^{\dagger})}$$
where $G_k$ is a complex matrix with independent Gaussian entries. \\
For each sample, we compute the Frobenius norm deviation
$$\epsilon_k=|| \mathcal{E}(\rho_k)\mathcal{E}(\sigma)-\mathcal{E}(\rho_k \sigma)||_F.$$
The results are presented as a discrete scatter plot in figure \ref{fig:violation}, where each point corresponds to one independently sampled state $\rho_k$. The vertical axis is plotted on a logarithmic scale. Here the deviations are consistently several orders of magnitude larger (typically $10^{-2}-10^{-1}$), with noticeable variation across samples.\\\\
These observations demonstrate that the identity in equation $\mathcal{E}(\rho)\mathcal{E}(\sigma)=\mathcal{E}(\rho \sigma)$ is not satisfied in general, even when $\sigma$ is non-degenerate and $\mathcal{E}$ is strictly incoherent. The violation arises from the non-Schur nature of the channel. This numerical evidence supports the conclusion that the multiplicative condition characterizes a strictly smaller class of operations, closely related to Schur multiplier channels, rather than the full set of strictly incoherent operations.
\begin{figure}[ht]
\centering
\includegraphics[width=0.45\textwidth]{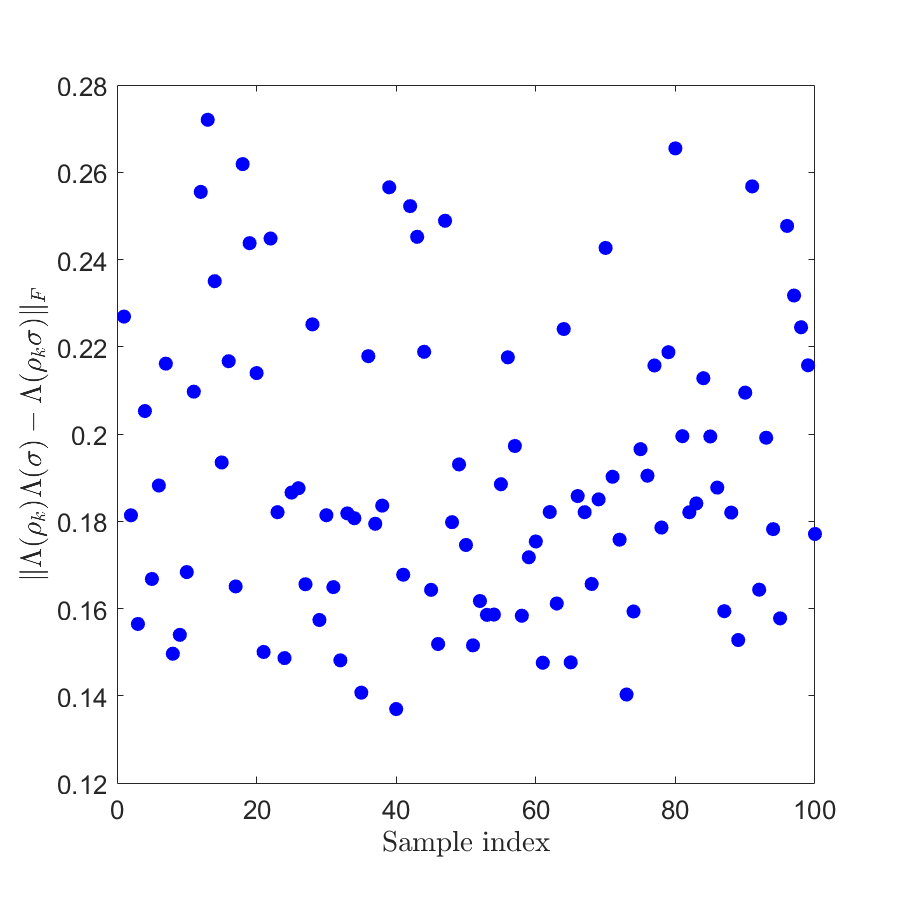}
\caption{Frobenius-norm deviation 
$\|\mathcal{E}(\rho_k)\mathcal{E}(\sigma)-\mathcal{E}(\rho_k\sigma)\|_F$ 
for randomly generated density matrices $\rho_k$. 
The incoherent state $\sigma=\mathrm{diag}(0.3,0.7)$ is non-degenerate. 
The SIO $\mathcal{E}$ is not of Schur-multiplier form. 
The nonzero deviations confirm violation of the multiplicative relation.}
\label{fig:violation}
\end{figure}
\section{Conclusion}
\label{sec:conclusion}
In this work, we investigated whether the use of catalysis can increase the coherence fraction of quantum states after passing through a quantum channel $\Lambda$. Specifically, we have used a catalysis $\tau_c$ to transfer the original input state $\rho_s$ to another state $\rho_s'$, without employing the catalysis $\tau_c$, which is undoubtedly impossible from a single copy of the state $\rho_s$. When $\rho_s'$ passes through the quantum channel $\Lambda$, we proved that the resulting state's coherence fraction $C_F(\Lambda(\rho_s'))$ can exceed $C_F(\Lambda(\rho_s))$. We have shown that this enhancement is possible under two conditions. First, when the quantum channel $\Lambda$ and the IO $\mathcal{E}$ satisfy the condition $C_F(\Lambda(\mathcal{E}(\rho_s)))= C_F(\mathcal{E}(\Lambda(\rho_s)))$ (see theorem \ref{th:main 2nd theorem}). Second, even if this condition is not fulfilled, the enhancement is still possible if $\Lambda$ is an incoherent channel (see theorem \ref{th:1st main theorem}). Here, we have also shown a practical application of our setup in a phase discrimination task. If the input state is somehow subject to decoherence or noise $\Lambda$, we use a correlated catalysis to prepare a processed state $\rho_s'$, and using $\Lambda(\rho_s')$ instead of $\Lambda(\rho_s)$ can lead us to better performance in phase discrimination. Most of the work on catalysis focuses on direct state transformations, but catalysis for channel downstream tasks like phase discrimination is unexplored. In our previous work \cite{char2024boosting}, we already proved correlated catalysts can improve phase discrimination when the input state is not affected by noise. This result extends that to noise-affected inputs for phase discrimination, which is more realistic for noisy quantum devices. In this work, we have provided an example of a quantum channel, viz., quantum addition channel, that commutes with a particular type of SIO (a subclass of IO), and hence the conclusion of theorem \ref{th:main 2nd theorem} is valid for this channel. We also provide another example illustrating theorem \ref{th:1st main theorem}, that for a qutrit phase-damping channel (an incoherent channel), correlated catalysis can help protect or even increase the coherence fraction. Consequently, this also leads to an example of improved performance in phase discrimination tasks when the input state is affected by phase-damping noise. In this work, we have also derived a necessary and sufficient condition for a CPTP map $\mathcal{E}$ with property $\mathcal{E(I)}\subseteq \mathcal{I}$, to be of a particular type of SIO. It turns out that multiplicativity condition $\mathcal{E}(\rho\sigma)=\mathcal{E}(\rho)\mathcal{E}(\sigma)$, where $\sigma$ is a non-degenerate incoherent state and $\rho$ is any state, is the necessary and sufficient condition for $\mathcal{E}$ to be an SIO in the form $\mathcal{E}(X)=U(A \odot X)U^{\dagger}$, with $A \succeq 0$ and $A_{ii}=1$ for all $i$ and U is a incoherent unitary operation. Numerical simulations (see figure \ref{fig:violation}) using 100 randomly generated states from the Ginibre ensemble further confirm that this multiplicativity condition fails in general when $\mathcal{E}$ does not belong to this specific class of SIO. Our results highlight the operational significance of catalysis in noise-affected states within coherence resource theory, and we hope this work will open new directions for improving quantum information processing tasks in noisy environments within other resource theories. 
\section*{ACKNOWLEDGMENT}
We thank Dr. Swapan Rana, Dr. Atanu Bhunia, and Indranil Biswas for the fruitful discussion related to this work. Priyabrata Char acknowledges the financial support of the Institute Postdoctoral Fellowship, IIT Guwahati, India.

\bibliography{main}
\bibliographystyle{apsrev4-1}

\end{document}